\documentclass{article}
\usepackage{spconf,amsmath,epsfig}
\usepackage{amsmath, amssymb, amsthm, graphicx, latexsym,indentfirst,mathtools, url, underscore, amsmath, dsfont, wasysym, amssymb, bm, algorithm, algorithmicx, algpseudocode, enumitem, caption, subcaption, float}

\newtheorem{theorem}{Theorem}

\newtheorem{lemma}{Lemma}

\newcommand{\bigop}[1]{\mathcal{O}_{\Bbb{P}}\left(#1\right)}
\newcommand{\bigo}[1]{\mathcal{O}\left(#1\right)}

\newcommand{\prob}[1]{\mathbb{P}\left(#1\right)}

\newcommand{\vertiii}[1]{{\left\vert\kern-0.25ex\left\vert\kern-0.25ex\left\vert #1 
    \right\vert\kern-0.25ex\right\vert\kern-0.25ex\right\vert}}
    
\algnewcommand\algorithmicinput{\textbf{Input:}}
\algnewcommand\Input{\item[\algorithmicinput]}
\algnewcommand\algorithmicmachine{\textbf{Machine:}}
\algnewcommand\Machine{\item[\algorithmicmachine]}
\algnewcommand\algorithmiccentralhub{\textbf{Central Hub:}}
\algnewcommand\CentralHub{\item[\algorithmiccentralhub]}

\title{Efficient Distributed Estimation of Inverse Covariance Matrices}
\name{Jes\'us Arroyo*, Elizabeth Hou* \thanks{*Both authors contributed equally}}
\address{Department of Statistics, University of Michigan} 

\begin{document}
\ninept
\maketitle
\begin{abstract}
In distributed systems, communication is a major concern due to issues such as its vulnerability or efficiency. In this paper, we are interested in estimating sparse inverse covariance matrices when samples are distributed into different machines. We address communication efficiency by proposing a method where, in a single round of communication, each machine transfers a small subset of the entries of the inverse covariance matrix. We show that, with this efficient distributed method, the error rates can be comparable with estimation in a non-distributed setting, and correct model selection is still possible. Practical performance is shown through simulations.
\end{abstract}
\section{Introduction}
\label{sec:intro}

\noindent The collection of copious and meticulous amounts of information has led to the modern phenomena of datasets being both high-dimensional and very large in sample size. These massive datasets are distributed over multiple machines due to size limitations or because data is collected and stored independently. Even in the case when a single machine is large enough, there are efficiency, security, and privacy concerns in aggregating all the data onto one machine. Bandwidth restrictions can make impossible or inefficient to send large amounts of data and the more communication in the system, the more vulnerable it is to attacks. In addition, the raw dataset may contain sensitive information in the individual samples such as in medical or financial records. Thus, it is advantageous, if not necessary, for each machine to be able to calculate compact estimates that can be efficiently communicated and preserve the confidentiality of the samples. 

Estimation of the inverse covariance matrix is used in the analysis of different types of data, such as gene expression or brain imaging. In particular, when samples are Gaussian, the non-zero entries of the inverse covariance correspond to the edges of a Gaussian Markov random field. Thus, accurately estimating the non-zero pattern of the inverse covariance matrix is crucial.

In this paper, we develop a method for estimating a sparse inverse covariance matrix when the samples of the data are distributed among machines. Our method is efficient in communication, in the sense that only a single round of communication between each machine and a central hub is sufficient, and the bandwidth required is small compared to the size of the matrix estimated. 

\subsection{Related Work}

\noindent High-dimensional estimation of inverse covariance matrices is usually addressed by $\ell_1$ penalized convex optimization \cite{banerjee2008model,friedman2008sparse}. In a distributed setting, general frameworks for convex optimization involve multiple rounds of communication between all machines \cite{Boyd:2011:DOS:2185815.2185816,duchi2012dual}, which can be expensive. For some $\ell_1$ penalization problems, the communication cost can be reduced to a single round of communication \cite{Zhang:2013:CAS:2567709.2567769,lee2015} by reducing the bias with bootstrapped or debiased estimators \cite{2013Javanmard}. Our work follows a similar approach as \cite{lee2015} for penalized linear regression; however, we introduce efficiency not only in the amount of communication, but also in the size of the communication channel. Related work on inverse covariance estimation has focused on the case where variables are distributed and the structure of the associated graphical model is known \cite{meng2014marginal}. Here, our samples are distributed, and we must estimate the structure of the matrix and the value of its entries.

\subsection{Outline}

\noindent The rest of this paper is organized as follows: Section 2 presents a communication efficient method for inverse covariance estimation in a distributed setting. Section 3 studies the error rates of the estimator and shows that model selection consistency is possible. Section 4 studies the practical performance and compares our method with other distributed and non-distributed approaches. Proofs of the theoretical results are included in the appendix. 

\section{Distributed Inverse covariance estimation by debiasing and thresholding}

\noindent We work in a setting where samples are distributed among $M$ different machines. Each observation is a vector of size $p$ coming from a distribution with covariance $\Sigma$ and inverse covariance $\Theta$. Let $ s $ be the number of non-zero entries in $\Theta$ and $ d $ the maximum number of non-zeros per row.  Let $\mathcal{S}$ denote the set of non-zero entries of $\Theta$ and $\mathcal{S}^c$ the set of zeros. 
We assume that the data is split equally over all machines, where each machine has $n$ observations, so we denote $X_m \in \mathbb{R}^{n \times p }$ as the data matrix on machine $m$. We are interested in estimating $ \Theta$. 

In a distributed setting, it is desirable to have a single round of communication between each machine and the central hub. Moreover, bandwidth or storage capacity often limits the amount of data that can be shared with the central hub, forcing the data to be distributed. Thus, our approach is based on constructing sparse estimators on each machine that when aggregated in the central hub, provide a good estimation for $\Theta$.

In the high dimensional setting, where $n \ll p$, a common approach to obtain a sparse estimator of $\Theta$ is by minimizing the $\ell_1$-penalized log-determinant Bregman divergence. Thus on each machine, this estimator, known as \emph{graphical lasso}, is defined as
\begin{equation} \label{l1-opt}
\hat{\Theta}_m := \underset{\theta \in \mathbb{S}^p_{++}}{\arg\min} \left\{ \,\text{tr}(\theta \hat{\Sigma}_m) - \log\det\theta + \lambda ||\theta||_{1, \text{off}}\right\},
\end{equation}
where $\|\cdot\|_{1,\text{off}}$ denotes the $\ell_1$ norm of the off-diagonal entries of a matrix, and $\mathbb{S}^p_{++}$ is the cone of positive definite matrices of size $p$.
For a single machine, this estimator has been studied and shown to be asymptotically consistent in mean squared error with rate $\bigop{s\log p /n}$ \cite{rothman2008}. Moreover, the set of non-zero entries of $\hat{\Theta}$ coincides with $\mathcal{S}$ when $\lambda\asymp\sqrt{\log p/n}$ and under certain conditions \cite{ravikumar2011}.

A naive approach for distributed estimation would be to average the estimated $\hat{\Theta}_m$ from each machine. However, this estimator is biased due to the $\ell_1$ penalty, and averaging only improves the variance, not the bias. We adopt a similar approach as \cite{lee2015} did for lasso regression by trading-off the bias for variance. 

The \emph{debiased graphical lasso} estimator was proposed in order to construct confidence intervals for the entries of $\Theta$ \cite{jankova2015}. The idea of this estimator is to invert the Karush-Kuhn-Tucker (KKT) conditions of the optimization problem in equation \eqref{l1-opt} in order to get a debiased estimator defined as
\begin{equation} \label{debiased est} \hat{\Theta}_m^d := \hat{\Theta}_m + \hat{\Theta}_m((\hat{\Theta}_m)^{-1}-\hat{\Sigma}_m)\hat{\Theta}_m.
\end{equation}
The \emph{debiased graphical lasso} has the appealing property that each entry of the matrix is asymptotically normal distributed. It is shown that $\hat{\Theta}_m^d$ can also be written as
\begin{equation} \label{loss}
\hat{\Theta}_m^d = \Theta-\Theta(\hat{\Sigma} - \Sigma)\Theta + \Delta_m,
\end{equation}
where $\|\Delta_m\|_\infty=\bigop{d\log p/n}$ accounts for the bias, and the second term in the equation is asymptotically normal \cite{jankova2015}. The bias of these estimators is of smaller order than the bias of the graphical lasso, and the variance is reduced when we average these estimators in the central hub to get an overall estimator $\frac{1}{M}\sum_{m=1}^M \hat{\Theta}^{d}_{m } $.
When the data is not distributed into too many machines, the averaged estimator can get similar error rates in $\ell_\infty$ norm as the graphical lasso performed on all the data (see Lemma \ref{lemma:avgglasso}).

The debiased graphical lasso estimator is not sparse. This fact is problematic in a distributed setting, since it will require the transfer of $p^2$ entries, which might be larger than the data on each machine. Under the sparsity assumption on $\Theta$, we are actually only interested in the value of $s+p$ entries. Hence, on each machine we select the most significant coefficients of $\hat{\Theta}_m$ and send them to the central hub. The sparse estimator is defined as
\begin{equation*}
\hat{\Theta}^{d,\rho}_{ij} := \hat{\Theta}^{d}_{ij} \, \mathbb{I}\left(|\hat{\Theta}^d_{ij}| > \rho \, \hat{\sigma}_{ij}\right),
\end{equation*}
where $\hat{\sigma}_{ij}$ is an estimator for $\sigma^2_{ij} = \text{Var}(\Theta_{i\cdot} X_{1\cdot} X_{1\cdot}^T \Theta_{\cdot j})$ and $\mathbb{I}(E)$ is the indicator of the event $E$. For gaussian random vectors, a good estimator for this quantity is $\hat{\sigma}^2_{ij}=\hat{\Theta}_{ii}\hat{\Theta}_{jj}+\hat{\Theta}_{ij}^2$ (Lemma 2 of \cite{jankova2015}). In order to achieve correct estimation (in the central hub) of the support of $\Theta$, it is optimal to send as many entries as possible. So we let the threshold parameter $\rho$ be a function of $B$, the bandwidth of the communication channel from each machine to the central hub, and we set $\rho$ as the smallest threshold that still fits in the channel. In Algorithm \ref{alg machine}, we summarize the estimation procedure on each machine.

\begin{algorithm} \label{alg machine}
\caption{Thresholded Debiased Estimator}
\begin{algorithmic} 
\Machine Do on each machine $m$
\Input Data matrix $X_{m}$, penalty $\lambda$, threshold $\rho$.
\State 1) Define $\hat{\Theta}$ as the solution of \eqref{l1-opt} with penalty $\lambda$.
\State 2) Define $ \hat{\Theta}^d = \hat{\Theta} + \hat{\Theta}(\hat{\Theta}^{-1}-\hat{\Sigma})\hat{\Theta} $.
\State 3) Calculate an estimate $\hat{\sigma}_{ij}^2$ for $\sigma_{ij}^2$.
\State 4) Threshold $\hat{\Theta}^d$, so $ \hat{\Theta}^{d, \rho}_{ij} = \hat{\Theta}^{d}_{ij} \, \mathbb{I}\left(|\hat{\Theta}^d_{ij}| > \rho \, \hat{\sigma}_{ij} \right) $. \\
\Return Sparse matrix $\hat{\Theta}^{d, \rho}$.
\end{algorithmic}
\label{alg machine}
\end{algorithm}

The average debiased graphical lasso estimator is also not sparse, which might be unpractical, in particular because estimating the set of non-zeros can be more important than the values of the entries themselves. However, if the averaged estimator is thresholded at a certain level $\tau\asymp \sqrt{\log p / (nM)}$, correct model selection on the non-zero entries of the matrix is possible (see Theorem \ref{theorem:rates}). Each entry again requires an estimator of $\sigma^2_{ij}$. For normally distributed data, we use $\bar{\Theta}_{ii}\bar{\Theta}_{jj} +\bar{\Theta}_{ij}^2$, where $\bar{\Theta}=\frac{1}{M}\sum_{m=1}^M \hat{\Theta}^{d, \rho}_m$. Algorithm \ref{alg hub} shows the complete procedure in the central hub.

\begin{algorithm} 
\label{alg hub}
\caption{Distributed Inverse Covariance Estimator}
\begin{algorithmic}
\CentralHub
\Input Estimators from each machine $\hat{\Theta}^{d, \rho}_m$, threshold $\tau$.
\State 1) Define $ \hat{\Theta}^D :=\hat{\Theta}^D(\rho)= \frac{1}{M}\sum_{m=1}^M \hat{\Theta}^{d, \rho}_{m} $.
\State 2) Calculate an estimate $\hat{\sigma}_{M, ij}^2$ for $\sigma_{ij}^2$.
\State 3) Threshold $\hat{\Theta}^D$, so $ \hat{\Theta}_{ij}^{D, \tau} = \hat{\Theta}_{ij}^D \, \mathbb{I} \left(|\hat{\Theta}^D_{ij}| > \tau \, \hat{\sigma}_{M, ij} \right) $.\\
\Return Overall distributed estimator $\hat{\Theta}^{D, \tau}$.
\end{algorithmic}
\label{alg hub}
\end{algorithm}

\section{Theoretical results}

\noindent In this section, we derive bounds for the estimation error of our method. We show that the error of our distributed estimator has the same rate as the non-distributed graphical lasso when the number of machines increases slower than $M \lesssim n/(d^2\log p)$. Moreover, we show that a bandwidth of size $O(p^{2-c})$, with $c$ an absolute constant, is enough to correctly identify the set of non-zero entries of $\Theta$.

The following assumptions are necessary in order for the graphical lasso and debiased graphical lasso to have a good estimation performance \cite{ravikumar2011,jankova2015}, and we require them for our theoretical results.
\begin{enumerate}[label=\textbf{(A\arabic*)}]
\item \label{a1} There exists some $ \alpha \in (0, 1] $ such that 
$ \underset{e \in \mathcal{S}^c}{\max} \, || \Gamma_{e\mathcal{S}} (\Gamma_{\mathcal{S}\mathcal{S}})^{-1} ||_1\le (1-\alpha) $
where $ \Gamma := (\Theta)^{-1} \otimes (\Theta)^{-1} $ is the Hessian of (\ref{l1-opt}).
\item \label{a2} There exists a constant $ L < \infty $ such that 
$ 1/L \leq \Lambda_{min}(\Sigma) \leq \Lambda_{max}(\Sigma) \leq L $
where $ \Lambda_{min}(\Sigma) $ and $ \Lambda_{max}(\Sigma) $ are the minimum and maximum eigenvalues of $ \Sigma $.

\item \label{a3} The samples of the data are subgaussian random variables $X\in\Bbb{R}^p$, with $\mathbb{E}[X]=0 $, $ \text{Cov}(X)=\Sigma$ and subgaussian norm $\|X\|_{\phi_2}\leq K$.

\item \label{a4} The quantities  $\vertiii{\Sigma}_\infty$ and $\vertiii{(\Gamma_{\mathcal{S}\mathcal{S}})^{-1}}_\infty$ are bounded, where $\vertiii{\cdot}$ is the opertor norm of a matrix.
\end{enumerate}

In the literature, it is common to allow the error rates to depend on the bounding constants from the previous assumptions. Here, in order to keep the results simple, we state the error rates as a function of the dimensionality ($n,p,M$), sparsity ($s$, $d$), and smallest entries of $\Theta$, while keeping the other quantities bounded by absolute constants.

The work from \cite{jankova2015} studies the error rate of debiased graphical lasso. This result can be extended to the average of multiple debiased estimators as follows.

\begin{lemma}\label{lemma:avgglasso}
Suppose that assumptions \ref{a1} through \ref{a4} hold, $M<p$ and define $\lambda_m \asymp\sqrt{\log p/n}$ in equation \eqref{l1-opt}. Then, the averaged debiased graphical lasso estimator satisfies
\begin{equation*}
\left\|\frac{1}{M}\sum_{m=1}^M\Theta_m^{d, \rho} - \Theta\right\|_\infty = \bigop{\max\left\{ \sqrt{\frac{\log p}{Mn}},d\frac{\log p}{n} \right\}}.
\end{equation*}
\end{lemma}

The previous lemma splits the error rate of the distributed estimator into two parts, which correspond to the variance and the bias. The variance vanishes as the number of machines increases, but the bias remains. However, when $M\lesssim n/(d^2\log p)$, the variance becomes dominant and the error is $\bigop{\sqrt{\log p/(Mn)}}$. This is the same error rate as the graphical lasso when the data is not distributed \cite{rothman2008}. Previous work shows similar results for distributed $\ell_1$ penalized regression \cite{lee2015}. Our method expands the framework to the graphical lasso estimator. Additionally, in the next theorem we show that  a bandwith of size $B\asymp p^{2-c}$ is enough
to select the correct set of entries and a similar rate for the mean squared error as the graphical lasso on the full data.

\begin{theorem} \label{theorem:rates}
Suppose that \ref{a1} through \ref{a4} hold and $M<p$. Define the tuning parameters of the algorithms as $\lambda\asymp\sqrt{\log(p)/n}$ and $\tau\asymp\sqrt{\log(p)/(Mn)}$. If $ \eta_\text{min} := \underset{(i,j)\in\mathcal{S}}{\min} \frac{|\Theta_{ij}|}{\sigma_{ij}} = \Omega(\sqrt{\log(p)/n})$, then there exists a constant $c\in(0,1]$ such that the following results hold for a bandwidth $B=\Omega\left(p^{2-c}\right)$.
\begin{enumerate}
\item Algorithm 2 recovers the correct set of non-zero entries of $\Theta$ with high probability, that is, \[\prob{\mathcal{S}(\hat{\Theta}_m^{D, \tau}) = \mathcal{S}(\Theta)}\geq 1-\bigo{1/p}.\]
\item The mean squared error of the estimator given by Algorithm 2 satisfies
\begin{equation} \label{mse}
\left\|\hat{\Theta}_m^{D, \tau} - \Theta\right\|_F^2 = \bigop{(s+p)\max\left\{\frac{\log p}{Mn},\frac{d^2\log^2 p}{n^2}\right\}}.
\end{equation}
\end{enumerate}
\end{theorem}

\section{Simulation results}

\noindent To evaluate the performance of our distributed estimator, we conduct a simulation study. We study the effect of varying the total sample size by changing the number of machines in the distributed system. We fix the number of variables to $p=1000$ and the sample size on each machine to $n=100$. Samples were generated from a normal distribution $N(0_{p}, \Sigma)$ such that $\Theta=(\Sigma)^{-1}$ has the form $\Theta_{i,i}=1$ and $\Theta_{i,i+1}=\Theta_{i,i-1}=0.4$ for $i=1,\ldots,p$. Thus, the associated Gaussian graph is a chain. All of our simulation results are calculated over 100 different trials.

To solve the problem \eqref{l1-opt}, we use  GLasso \cite{friedman2008sparse} with the R package \texttt{huge} \cite{zhao2012huge}. We set the tuning parameters $\lambda$ and $\tau$ according to the rates in Theorem \ref{theorem:rates}, so for each machine $\lambda=\sqrt{\log(p)/n}$ and $\tau=\sqrt{\log(p)/(Mn)}$. The bandwidth is set to $B=10p$, so only 1\% of the entries of $\hat{\Theta}^d_m$ are sent to the central hub. 

We compare the performance of our distributed estimator (Distributed) with the following estimators.
\begin{enumerate}
\item (Naive) An estimator based on averaging the graphical lasso estimators from each machine $ \hat{\Theta}^{\text{naive}} = \frac{1}{M}\sum_{m=1}^M \hat{\Theta}_{m} $.
\item (Full) An estimator based on the full non-distributed data given by the graphical lasso $ \hat{\Theta}^{\text{full}}$.
\item (Full Debiased) Since debiasing decreases the bias of the estimation, we also compare with a thresholded debiased graphical lasso estimator on the full data $\hat{\Theta}^{D,\text{full}}$.
\end{enumerate}

\begin{figure}[h]
\centering
\begin{minipage}[b]{.7\linewidth}
  \centering 
  \centerline{\includegraphics[width=\linewidth]{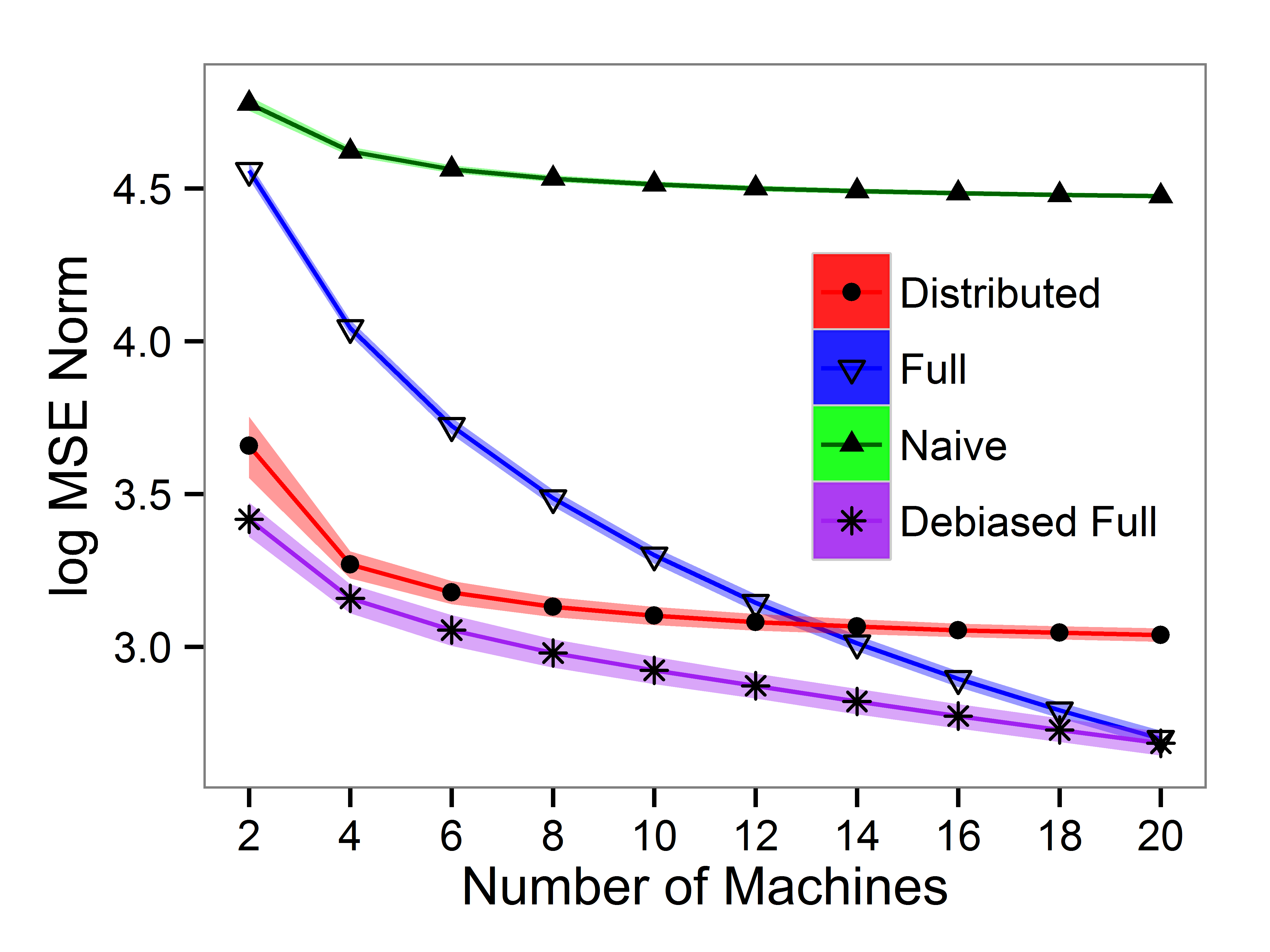}}
\end{minipage}
\begin{minipage}[b]{.7\linewidth}
 \centering
 \centerline{\includegraphics[width=\linewidth]{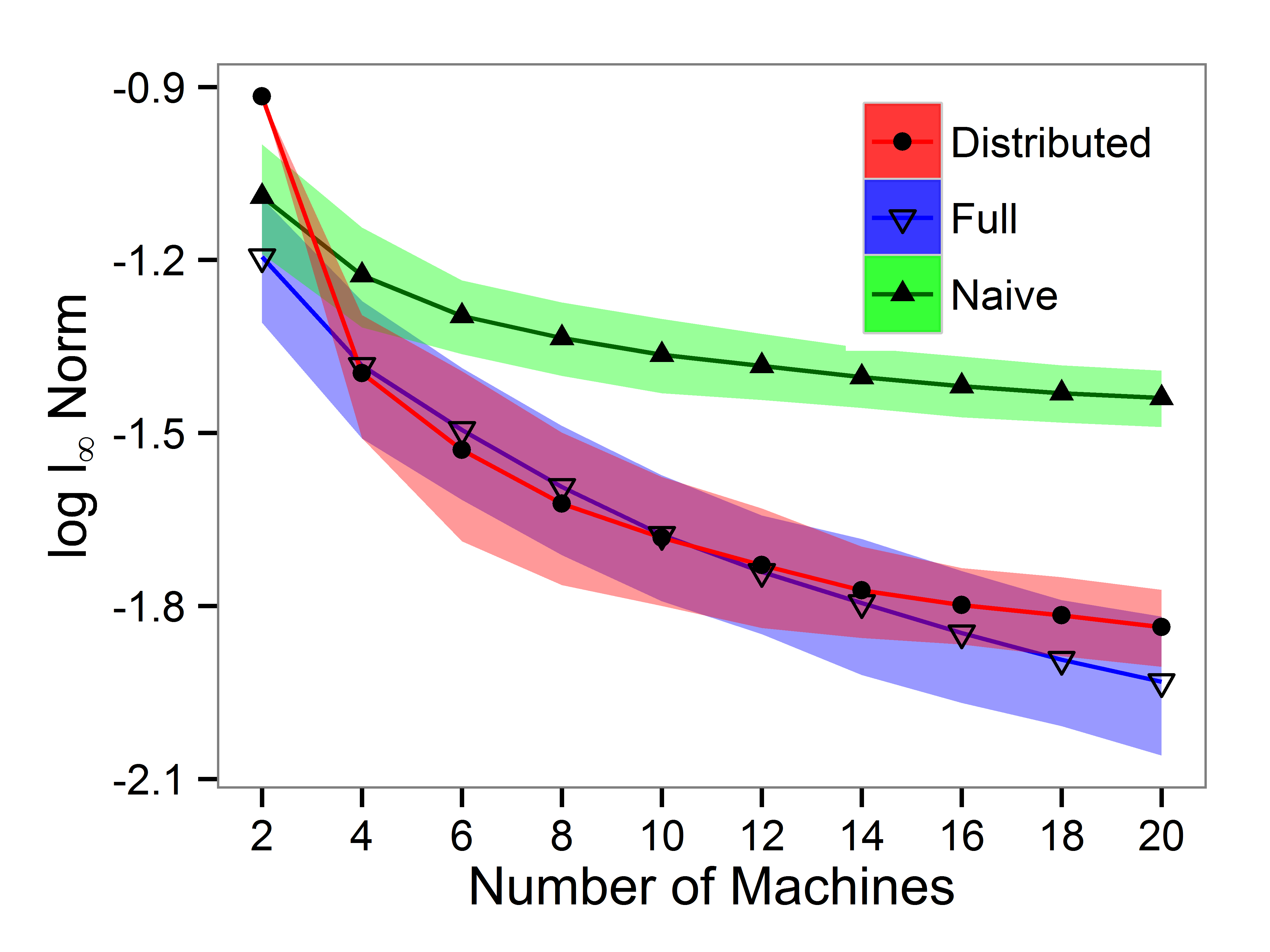}}
 \end{minipage}
\caption{Estimation error of the different estimators based on MSE and $\ell_\infty$ norm. Note that Debiased Full has the same error in $\ell_\infty$ as the Full estimator. For a small number of machines, the performance of our distributed estimator is similar to that of the estimator on the full data, and superior to the naive estimator.}
\label{fig:errorplots}
\end{figure}

\begin{figure}[H]
\begin{minipage}[b]{.5\linewidth}
  \centering
  \centerline{\includegraphics[width=\linewidth]{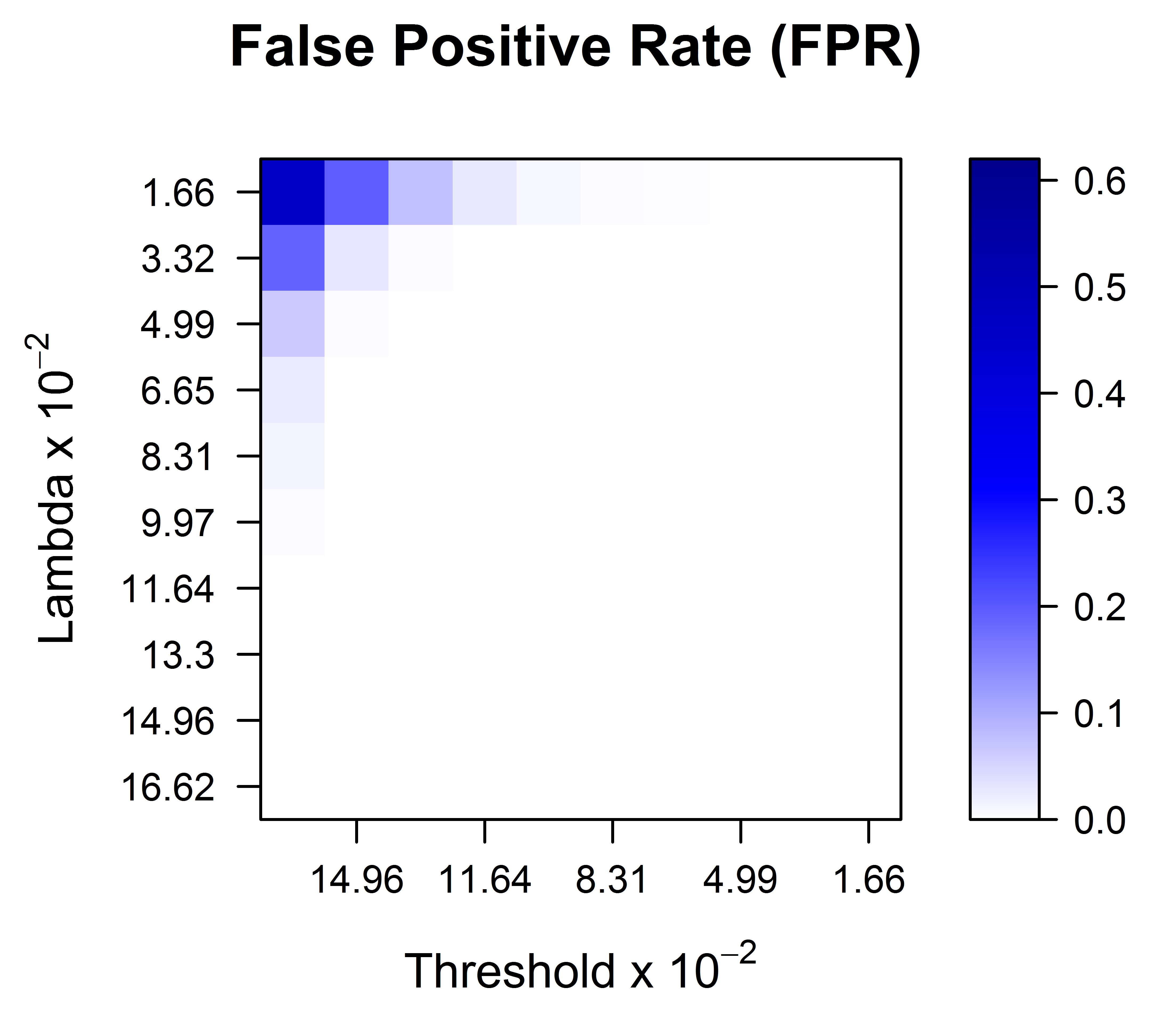}}
\end{minipage}
\hfill
\begin{minipage}[b]{.5\linewidth}
  \centering
  \centerline{\includegraphics[width=\linewidth]{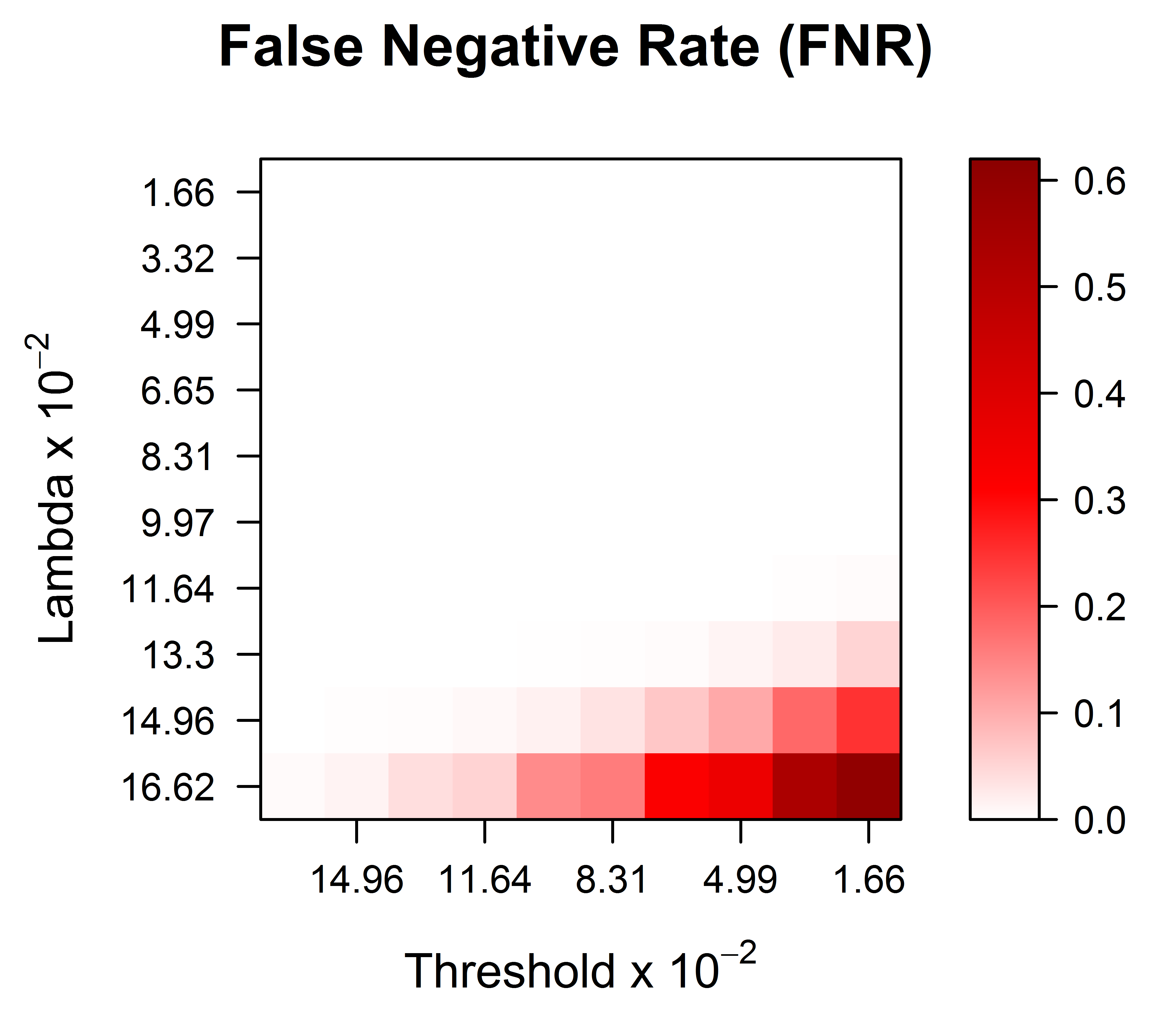}}
\end{minipage}
\caption{FPR and FNR  for recovering the correct set of non-zero entries of $\Theta$ using our distributed estimator. The number of machines is set to $10$ and we vary the tuning parameters of the algorithm. We note the robustness of the estimators in recovering the correct support using different values of $\lambda$ and $\tau$. }
\label{fig:modelselection}
\end{figure}

In Figure \ref{fig:errorplots}, the errors of the different estimators are shown.  In general, a better estimator can be obtained using the full data, which is expected. However, when the number of machines is small, the performance of our distributed estimator is comparable to the estimators on the full data, both in mean squared error and $\ell_\infty$ norm. In mean squared error, when $M$ is small, the debiased estimators perform better because they trade off bias for variance. Thus, Debiased Full performs the best and Distributed has similar performance. Naive performs poorly under this norm for all $M$. In $\ell_\infty$ norm, the performance of Full and Full Debiased is exactly the same because the largest error is from the entries in the diagonal, which are equal since the diagonal is not penalized. The error of Distributed is approximately the same for a wide range of values of $M$. However, in $M=2$, the performance of Distributed is affected by the bandwidth since missing edges in any machine have a larger impact on the error. But as the number of machines increases, the estimation of Distributed becomes insensitive to the bandwidth.

To evaluate the robustness of the method against the selection of the tuning parameters $\lambda$ and $\tau$, we measure false positive rate (defined as percentage of zeros of $\Theta$ identified as edges by the estimator) and false negative rate (defined as percentage of edges missed by the estimator). We use the same settings as the previous scenario, but the number of machines is fixed to $M=10$. The values of the tuning parameters for our method vary proportionally to $\lambda = \beta\sqrt{\log(p)/n}$ and $\tau = \beta\sqrt{\log(p)/(Mn)}$ with $\beta$ between $0.2$ and $2$. The bandwidth is still fixed at 1\% of the entries. We observe that for a wide range of parameters we recover the correct set of non-zeros. 

\section{Discussion}

\noindent We have proposed a method for estimating sparse inverse covariance matrices when the samples of a dataset are distributed over different machines. Our method agrees with other results for efficient distributed estimation in high-dimensional settings, and we also introduced efficiency in the bandwidth size. Asymptotically, the performance of our estimator is analogous to estimators with samples that are not distributed. Our simulation results are consistent with the theoretical rates and also show that we perform significantly better than a naive approach to estimation in a distributed setting. 

\section{Appendix}
\begin{proof}[Proof of Lemma \ref{lemma:avgglasso}]
Using the decomposition given in equation \eqref{loss}, we have
\begin{flalign*}
\left\|\sum_{m=1}^M \frac{\Theta_m^d}{M} -\Theta\right\|_\infty &\leq \left\| \Theta\left(\sum_{m=1}^M\frac{\hat{\Sigma}_m}{M}-\Sigma\right)\Theta\right\|_\infty + \sum_{m=1}^M\left\|\frac{\Delta_m}{M}\right\|_\infty\\
&  =  \left\|\Theta(\hat{\Sigma}- \Sigma)\Theta\right\|_\infty + \frac{1}{M}\sum_{m=1}^M\left\|\Delta_m\right\|_\infty,
\end{flalign*}
where $\hat{\Sigma}$ is the covariance matrix using the full data. Thus, the first term can be bounded as $ \left\|\Theta(\hat{\Sigma}- \Sigma)\Theta\right\|_\infty = \bigop{ \sqrt{\log p/(Mn) } } $ (equation (11) of \cite{jankova2015}) and the second term can be bounded as 
\begin{flalign*}
& \mathbb{P}\left(\frac{1}{M}\sum_{m=1}^M\left\|\Delta_m\right\|_\infty \leq \delta \right) \geq \mathbb{P}\left( \left\|\Delta_m\right\|_\infty \leq \delta \right)^M \\
& = (1 - 1/p^2)^M > (1 - 1/p^2)^p > 1 - 1/p
\end{flalign*}
when $ M < p $ and $\delta \asymp \sqrt{\log p / n}$  \cite{ravikumar2011}.
\end{proof}

\begin{lemma} \label{lemma:MSprob}
Suppose assumptions \ref{a1} through \ref{a4} hold. Let  $\gamma$ be a constant in $[0,1]$ and define the events
\begin{eqnarray*}
\mathcal{E}_\gamma(\hat{\Theta}^{D}(0)) & = & \left\{|\hat{\Theta}^{D}_{ij}(0)|>\gamma \sigma_{ij}\eta_\text{min}\text{ for }(i,j)\in\mathcal{S} \right\},\\
\mathcal{N}_{\gamma,b}(\hat{\Theta}^{D}(0)) & =& \left\{\sum_{(i,j)\in\mathcal{S}^C} \mathbb{I} \left(|\hat{\Theta}^D_{ij}(0)|>\gamma \sigma_{ij}\eta_\text{min}\right)\leq b\right\}.
\end{eqnarray*}
Thus, $\mathcal{E}_\gamma$ denotes the event that all non-zero entries of $\Theta$ in $\hat{\Theta}^D(0)$ are above a threshold $\gamma \sigma_{ij}\eta_\text{min}$, and $\mathcal{N}_{\gamma,b}$ the event when at most $b$ entries of $\hat{\Theta}^D(0)$ corresponding to the zeros of $\Theta$ are above the threshold. Then, 
\begin{eqnarray*} \label{probInt}
\prob{\mathcal{E}_\gamma} &= &1 -
 \bigo{s e^{-0.5 (1-\gamma)^2 n M ( \eta_\text{min} )^2 }}\\
\prob{\mathcal{N}_{\gamma,b}} & = & 1 -
 \bigo{(p^2-s-b) e^{-0.5 \gamma^2 n M \left( \eta_\text{min} \right)^2 }}\\\end{eqnarray*}
\end{lemma}

\begin{proof}[Proof of Lemma \ref{lemma:MSprob}]
By the asymptotic normality of the debiased graphical lasso entries \cite{jankova2015},  $\hat{\Theta}^D_{ij}(0)$ is asymptotically distributed as $N(\Theta_{ij}, \sigma_{ij}^2/(nM))$. Therefore,
\begin{eqnarray} \label{probE}
\prob{\mathcal{E}_\gamma^c} & \leq & \notag \sum_{(i,j)\in\mathcal{S}}\prob{|\hat{\Theta}^D_{ij}(0)|<t_{ij}}+o(1) \notag \\
& \leq & s\max_{(i,j)\in\mathcal{S}}\prob{|Z+\sqrt{nM}\frac{\Theta_{ij}}{\sigma_{ij}}|< \gamma \sqrt{nM} \eta_{\text{min}}} +o(1)\notag \\
&\leq & s\prob{Z< -(1-\gamma)\sqrt{nM}\eta_{\text{min}}}+o(1),\notag
\end{eqnarray}
where $Z$ is a standard normal random variable. Using a similar argument, the second event can be bounded as
\begin{eqnarray} \label{probN}
\prob{\mathcal{N}_{\gamma,b}^c} & = & \prob{\bigcup_{\substack{A\subset S^c \\ |A|\leq b}}\bigcap_{(i,j)\in A^c}(|\hat{\Theta}_{ij}^D(0)|\geq t_{ij})} +o(1)\notag \\
& \leq & \min_{\substack{A\subset S^c \\ |A|\leq b}} \sum_{(i,j)\in A^c}\prob{|Z|\geq\gamma\sqrt{nM}\eta_{\text{min}}}+o(1) \notag \\
&\leq & 2(p^2-s-b)\prob{Z\leq-\gamma\sqrt{nM}\eta_{\text{min}}}+o(1)\notag
\end{eqnarray}
Using tails of a normal distribution, the results follow.
\end{proof}

\begin{proof}[Proof of Theorem \ref{theorem:rates}]
\emph{Part 1}. In order to recover the correct set of non-zero entries, the event $\mathcal{E}_{\tau}(\hat{\Theta}^{D}(\rho))\cap\mathcal{N}_{\tau,B-s}(\hat{\Theta}^{D}(\rho))$ must hold. Note that if in Algorithm 1, the threshold given by the bandwidth contains all non-zero entries, then $\hat{\Theta}^{D}_{ij}(\rho)=\hat{\Theta}^D_{ij}(0)$, for $(i,j)\in\mathcal{S}$. Moreover, 
since $\prob{|\hat{\Theta}^{D}_{ij}(\rho)|>h} \leq \prob{|\hat{\Theta}^{D}_{ij}(0)|>h}$ for $(i,j)\in\mathcal{S}^C$ and all $h>0$, 
then, by conditioning on the event $\cap_{m=1}^M\left(\mathcal{E}_{\rho}(\hat{\Theta}^d_m))\cap\mathcal{N}_{\rho,b}(\hat{\Theta}^d_m)\right)$, it holds that
\begin{flalign*}
&\prob{\mathcal{E}_{\tau}(\hat{\Theta}^{D}(\rho))\cap\mathcal{N}_{\tau,0}(\hat{\Theta}^{D}(\rho))}  \geq \\
&\prob{\mathcal{E}_{\tau}(\hat{\Theta}^{D}(0))\cap\mathcal{N}_{\tau,0}(\hat{\Theta}^{D}(0))}\prob{\cap_{m=1}^M\left(\mathcal{E}_{\rho}(\hat{\Theta}^d_m)\cap\mathcal{N}_{\rho,b}(\hat{\Theta}^d_m)\right)}.
\end{flalign*}

Using the conditions of Theorem \ref{theorem:rates}, in particular the size of $\eta_{\text{min}}$, and by Lemma \ref{lemma:MSprob}, it holds that the first term of the product is $1-\bigo{p^{-1}}$. Similarly, using Lemma \ref{lemma:MSprob} again, we calculate the probability of correct recovery in all machines 
\begin{eqnarray*}
\prob{\cap_{m=1}^M\left(\mathcal{E}_{\rho}(\hat{\Theta}^d_m)\cap\mathcal{N}_{\rho,B-s}(\hat{\Theta}^d_m)\right)} & = & \left(1-\bigo{\frac{1}{p^{1+c_1}}}\right)^M,
\end{eqnarray*}
which is $1-\bigo{\frac{1}{p}}$ when $M<p$, and $c_1>0$ is an absolute constant. Thus, correct model selection holds with high probability.

\noindent\emph{Part 2}.  Using the equivalence between $\ell_\infty$ and Frobenius norm together with the correct recovery result of part 1, we obtain equation \ref{mse} as a consequence of multiplying the result of Lemma \ref{lemma:avgglasso}  by $p+s$.
\end{proof}

\newpage

\section{Acknowledgments}
\noindent The research in this paper was partially supported by the Consortium for Verification Technology under Department of Energy National Nuclear Security Administration award number DE-NA0002534.

\bibliographystyle{IEEEbib}
\bibliography{strings,refs}

\end{document}